\documentclass[a4paper,USenglish,cleveref,autoref,thm-restate]{lipics-v2021}

\title{Parameterized Complexity of Finding Dissimilar Shortest Paths} 

\author{Ryo Funayama}{Graduate School of Information Science and Technology, Hokkaido University, Japan}{funayama.ryo.o3@elms.hokudai.ac.jp}{}{}

\author{Yasuaki Kobayashi}{Faculty of Information Science and Technology, Hokkaido University, Japan}{koba@ist.hokudai.ac.jp}{https://orcid.org/0000-0003-3244-6915}{JSPS KAKENHI Grant Numbers JP20H00595 and JP23H03344}

\author{Takeaki Uno}{National Institute of Informatics, Japan}{uno@nii.ac.jp}{}{}

\authorrunning{R. Funayama et al.} 

\Copyright{Ryo Funayama, Yasuaki Kobayashi, and Takeaki Uno} 



\category{} 

\relatedversion{} 

\nolinenumbers 
\usepackage{mathtools}
\usepackage{tcolorbox}

\newcommand{\dist}{{\rm dist}}

\newcommand{\xor}{\mathbin{\triangle}}

\newcommand{\optf}{{\rm path}}
\newcommand{\minbp}{{\rm mbp}}
\newcommand{\bp}{{\rm bp}}

\newcommand{\gad}{{\rm gad}}
\newcommand{\True}{{\tt true}}
\newcommand{\False}{{\tt false}}

\crefname{observation}{Observation}{Observations}

\hideLIPIcs

\EventEditors{John Q. Open and Joan R. Access}
\EventNoEds{2}
\EventLongTitle{42nd Conference on Very Important Topics (CVIT 2016)}
\EventShortTitle{CVIT 2016}
\EventAcronym{CVIT}
\EventYear{2016}
\EventDate{December 24--27, 2016}
\EventLocation{Little Whinging, United Kingdom}
\EventLogo{}
\SeriesVolume{42}
\ArticleNo{23}

\begin{document}

\maketitle

\begin{abstract}
We consider the problem of finding ``dissimilar'' $k$ shortest paths from $s$ to $t$ in an edge-weighted directed graph $D$, where the dissimilarity is measured by the minimum pairwise Hamming distances between these paths.
More formally, given an edge-weighted directed graph $D = (V, A)$, two specified vertices $s, t \in V$, and integers $d, k$, the goal of \textsc{Dissimilar Shortest Paths} is to decide whether $D$ has $k$ shortest paths $P_1, \dots, P_k$ from $s$ to $t$ such that $|A(P_i) \xor A(P_j)| \ge d$ for distinct $P_i$ and $P_j$.
We design a deterministic algorithm to solve \textsc{Dissimilar Shortest Paths} with running time $2^{O(3^kdk^2)}n^{O(1)}$, that is, \textsc{Dissimilar Shortest Paths} is fixed-parameter tractable parameterized by $k + d$.
To complement this positive result, we show that \textsc{Dissimilar Shortest Paths} is W[1]-hard when parameterized by only $k$ and paraNP-hard parameterized by $d$.
\end{abstract}

\section{Introduction}
The shortest path problem is one of the most fundamental problems, which is of great importance in both theoretical and practical points of view.
This problem has a wide range of applications in many real-world problems, such as navigation systems, supply chain management, and social network analysis.

In solving such real-world problems, we usually describe these problems as mathematical models, and then apply optimization algorithms to obtain solutions.
However, due to the complex nature of real-world problems, we may overlook several intricate constraints and/or preferences, which are difficult to accurately incorporate into mathematical models.
Moreover, we may even ignore these intricate factors in order to simplify mathematical models, allowing us to solve optimization problems efficiently.
Given this circumstance, optimization algorithms may not generate a good solution for original real-world problems, even if they generate an optimal solution for their mathematical models.

To circumvent this issue, finding \emph{diverse} solutions is an intuitive approach.
In this approach, algorithms are anticipated to generate a moderate number of high-quality and \emph{diverse} solutions, where by diverse solutions, we mean solutions that are ``sufficiently different'' from each other.
With such diverse solutions available, users can then select a desired solution by fully leveraging their cognitive abilitie.
Consequently, solution diversity would broaden the users' options, which is crucial within this approach.

There are numerous studies aimed at finding diverse solutions for various combinatorial problems.
In particular, several theoretical investigations have emerged recently from the perspectives of polynomial-time solvability~\cite{deBergMS23:arXiv:Finding,HanakaKKLO22:AAAI:Computing}, approximability~\cite{DoGN023,GaoGMKTTY22,HanakaKKKKO22:arXiv:Framework}, and fixed-parameter tractability~\cite{BasteFJMOPR22:AI:Diversity,BasteJMPR19:Algorithms:FPT,EibenLW23:arXiv:Determinantal,FominGJPS20,FominGPP021:STACS:Diverse,HanakaKKO21:AAAI:Finding,MerklPS23:ICDT:Diversity}.
Among these, the problem of finding diverse short(est) paths is particularly intriguing~\cite{Akgun:finding:2000,ChondrogiannisB20:VLDB:Finding,HackerBCA21SIGSPATIAL:Most,HanakaKKLO22:AAAI:Computing,Liu:Finding:2018,Voss:heuristic:2015,ZhengYYW07:INFOCOM:Finding}.

Given an unweighted directed graph $D$ with two specified vertices $s$ and $t$, and a positive threshold $\theta$, the problem of determining whether $D$ has two arc-disjoint paths $P_1$ and $P_2$ from $s$ to $t$ with length at most $\theta$ is known to be NP-complete~\cite{LiMS90}.
Hanaka et al.~\cite{HanakaKKLO22:AAAI:Computing} mentioned that finding two (not necessarily arc-disjoint) paths $P_1$ and $P_2$ from $s$ to $t$ with length at most $\theta$ that maximize $|A(P_1) \xor A(P_2)|$ is NP-hard.
Here, $A(P)$ denotes the set of arcs in a path $P$.
Thus, it is reasonable to restrict ourselves to consider finding diverse \emph{shortest} paths.
Zheng et al.~\cite{ZhengYYW07:INFOCOM:Finding} studied the problem of finding $k$ shortest paths $P_1, \dots, P_k$ from $s$ to $t$ in an arc-weighted directed graph $D$ minimizing
\begin{align*}
    \sum_{1 \le i < j \le k} |A(P_i) \cap A(P_j)|.
\end{align*}
They gave a polynomial-time algorithm by reducing the problem to the minimum-cost flow problem.
This objective function is equivalent to maximizing
\begin{align}\label{eq:maxsum}
    \sum_{1 \le i < j \le k} |A(P_i) \xor A(P_j)| = \sum_{1 \le i < j \le k}(|A(P_i)| + |A(P_j)| -|A(P_i) \xor A(P_j)|)
\end{align}
when all paths have an equal number of arcs.
Hanaka et al.~\cite{HanakaKKLO22:AAAI:Computing} gave a similar reduction to compute $k$ shortest paths $P_1, \dots, P_k$ directly maximizing~(\ref{eq:maxsum}), where each path may have a different number of arcs.
However, the objective function (\ref{eq:maxsum}) has a potential drawback: some paths found by the algorithm can be very similar or even identical.
A natural way to overcome this drawback is to consider the following objective function: finding $k$ shortest paths $P_1, \dots, P_k$ maximizing
\begin{align*}
     \min_{1 \le i < j \le k} |A(P_i) \xor A(P_j)|.
\end{align*}

In this paper, we work on this max-min counterpart of the above diverse shortest paths problem, which is formalized as follows.
\begin{tcolorbox}
\begin{description}
  \item[Problem.] \textsc{Dissimilar Shortest Paths}
  \item[Input.] An arc-weighted directed graph $D = (V, A, \ell)$ with $\ell\colon A \to \mathbb R_{>0}$, $s, t \in V$, and nonnegative integers $k, d \in \mathbb N$.
  \item[Task.] Decide if there are $k$ shortest paths from $s$ to $t$ in $D$ such that $|A(P_i) \xor A(P_j)| \ge d$ for $1 \le i < j \le k$.
\end{description}
\end{tcolorbox}

We show the following parameterized upper and lower bound results for this problem.
\begin{theorem}
    \textsc{Dissimilar Shortest Paths} can be solved deterministically in time $2^{O(3^{k}dk^2)}n^{O(1)}$, that is, it is fixed-parameter tractable parameterized by $k + d$.
\end{theorem}

\begin{theorem}\label{thm:W[1]-hardness}
    \textsc{Dissimilar Shortest Paths} is W[1]-hard parameterized by $k$, even on unweighted directed acyclic graphs whose underlying undirected graphs have pathwidth at most~$4$.
\end{theorem}

\begin{theorem}\label{thm:np-hardness}
    \textsc{Dissimilar Shortest Paths} is NP-hard for some constant $d$ even on unweighted directed graph.
\end{theorem}

\Cref{thm:W[1]-hardness,thm:np-hardness} indicate that the simultaneous parameterization of $k$ and $d$ seems to be vital to have the fixed-parameter tractability of \textsc{Dissimilar Shortest Paths}.

To show the algorithmic result, we adopt the strategy of \cite{FominGPP021:STACS:Diverse}, where they design a fixed-parameter algorithm for computing diverse perfect matchings.
In this strategy, we first compute dissimilar paths in a greedy manner.
If we successfully compute such a set of $k$ dissimilar paths, we are done.
Otherwise, by appropriately setting the dissimilarity in the greedy computation, the solution space (i.e., the collections of shortest paths) can be partitioned into at most $k$ ``balls with bounded radii''.
This partition enables us to compute a set of dissimilar paths in a ball by the color-coding technique~\cite{AlonYZ95}.

It is worth noting that a typical argument to prove the hardness of finding diverse solutions is refused.
For instance, the problem of finding diverse perfect matchings is NP-hard~\cite{FominGJPS20} as it is already NP-hard to find a set of \emph{disjoint} solutions~\cite{Holyer81:SICOMP:NP}.
A similar argument is used to show the NP-hardness of finding diverse global minimum cuts and diverse interval schedulings~\cite{HanakaKKKKO22:arXiv:Framework}.
However, it is well known that the problem of finding disjoint shortest paths from $s$ to $t$ in a directed graph is solvable in polynomial-time using standard network flow algorithms.
This fact requires nontrivial ideas for proving our lower bound results~\Cref{thm:W[1]-hardness,thm:np-hardness}.

We would like to mention that Eiben et al.~\cite{EibenLW23:arXiv:Determinantal} recently give a general algebraic technique to solve many diverse versions of combinatorial problems.
In fact, their technique is applicable to our problem \textsc{Dissimilar Shortest Paths}, which yields a \emph{randomized} $2^{O(dk^2)}n^{O(1)}$-time algorithm.
In contrast to this, our algorithm is \emph{deterministic}, consisting of elementary dynamic programming algorithms with a stadard tool.

\section{Preliminaries}

In this paper, we assume that the readers are familiar with standard terminology in graph theory and parameterized complexity.

For sets $X, Y$, we denote by $X \xor Y$ the symmetric difference of $X$ and $Y$, that is, $X \xor Y = (X \setminus Y) \cup (Y \setminus X)$.
It is easy to check that the following triangle inequality hold.
\begin{observation}\label{obs:triangle-inequality}
    For sets $X, Y, Z$, it holds that $|X \xor Z| \le |X \xor Y| + |Y \xor Z|$.
\end{observation}

We use the following identity.
\begin{observation}\label{obs:partition}
    For sets $W, X, Y, Z$ with $W \cap X = \emptyset$ and $Y \cap Z = \emptyset$, it holds that
    \begin{align*}
        |(W \cup Z) \xor (Y \cup Z)| = |W \xor Y| + |X \xor Y| + |W \xor Z| + |X \xor Z| - |W \cup X| - |Y \cup Z|.
    \end{align*}
\end{observation}
\begin{proof}
    Let $A = W \cup X$. Then we have
    \begin{align*}
        |A \xor (Y \cup Z)| &= |(A \setminus (Y \cup Z))| + |(Y \cup Z) \setminus A|\\
        &= |A| - |A \cap (Y \cup Z)| + |(Y \cup Z) \setminus A|\\
        &= |A| - |A \cap Y| - |A \cap Z| + |Y \setminus A| + |Z \setminus A|\\
        &= |A| + (|Y \setminus A| + |A| - |A \cap Y|) + (|Z \setminus A| + |A| - |A \cap Z|) - 2|A|\\
        &= |A \xor Y| + |A \xor Z| - |A|.
    \end{align*}
    By applying the same argument to $|A \xor Y|$ and $|A \xor Z|$, the observation follows.
\end{proof}

\paragraph*{Graphs}
Let $D = (V, A)$ be a directed graph.
We denote the vertex set and the arc set of $D$ by $V(D)$ and $A(D)$, respectively.
For $X \subseteq V$, we denote by $D[X]$ the subgraph of $D$ induced by $X$.
For $v \in V$, we denote by $\delta^-_D(v)$ the set of incoming arcs to $v$.
For two vertices $u$ and $v$, a directed path from $u$ to $v$ is called an \emph{$uv$-path}.
For an $uv$-path $P$ and an arc $(v, w)$, the $uw$-path obtained by concatenating $P$ and $(v, w)$ is denoted by $P + (v, w)$.

\paragraph*{Perfect hash family}
We use the derandmomized version of the famous color-coding technique due to \cite{AlonYZ95}.
Let $n$ and $k$ be positive integers.
We denote by $[n]$ the set of positive integers between $1$ and $n$.
A family $\mathcal F$ of functions $f\colon [n] \to [k]$ is called an \emph{$(n, k)$-perfect hash family} if for every subset $S \subseteq [n]$ with $|S| = k$, there is a function $f \in \mathcal F$ such that for $i, j \in S$ with $i \neq j$, it holds that $f(i) \neq f(j)$.
\begin{theorem}[\cite{NaorSS95:FOCS:Splitters}]\label{thm:perfect}
    For positive integers $n, k$, there is an $(n, k)$-perfect hashing family $\mathcal F$ of size $e^kk^{O(\log k)}$.
    Moreover, such a family $\mathcal F$ can be computed in time $e^kk^{O(\log k)}n\log n$.
\end{theorem}

\section{FPT algorithm}
In this section, given an edge-weighted directed graph $D = (V, A)$ with $\ell \colon A \to \mathbb R_+$, vertices $s, t \in V$, and integer $k, d$, we give an fixed-parameter tractable algorithm for \textsc{Dissimilar Shortest Paths}.
For notational convenience, we may regard $D$ as its edge set.
Thus, $|D|$ indicates the number of arcs in it.

\subsection{Preprocessing}
We apply the following simple preprocessing, which is crucial for subsequent phases.
In this preprocessing, we delete all edges that do not belong to any shortest path from $s$ to $t$ in $D$.
This can be done by first computing the distance labeling $\dist$ from $s$ and deleting edge $(u, v) \in A$ if $\dist(s, u) \neq \dist(s, v) + \ell((u, v))$, where $\dist(s, u)$ is the shortest path length from $s$ to $u$.
We also remove all vertices that are not reachable from $s$ or not reachable to $t$.
The directed graph obtained by applying this preprocessing is denoted by $D'$.
Since $\ell(e) > 0$ for all $e \in A$, $D'$ is indeed acyclic.
The following observation is immediate.

\begin{observation}\label{obs:preprocessing}
    Every shortest path from $s$ to $t$ in $D$ appears as a path from $s$ to $t$ in $D'$.
    Conversely, every path from $s$ to $t$ in $D'$ appears as a shortest path from $s$ to $t$ in $D$.
    Moreover, $s$ and $t$ are unique source and sink vertices of $D'$, respectively. 
\end{observation}

The directed graph $D'$ can be computed in polynomial time from $D$.
By this observation, we can assume that the input graph $D$ is acyclic and satisfying the properties in \Cref{obs:preprocessing} in the rest of this section,.

\subsection{An overview of our algorithm}
We first give an outline of our algorithm, which is used in \cite{FominGPP021:STACS:Diverse}.
To this end, we need to implement two subroutines, whose details are postponed to the next two subsections, and hence use them as a black box in this subsection.

For $u, v \in V$, we denote by $\mathcal P(u, v)$ the set of all paths from $u$ to $v$ in $D$.
By~\Cref{obs:preprocessing}, $\mathcal P(s, t)$ consists of all shortest paths from $s$ to $t$ in the original graph.
In the following, we assume the following two subroutines.

\begin{lemma}\label{lem:farthest-path}
    Given a non-negative integer $q$ and $r$ paths $P_1, \dots, P_r \in \mathcal P(s, t)$, there is an algorithm with running time $2^{O(r\log q)}n^{O(1)}$ that decides whether there is a path $P \in \mathcal P(s, t)$ such that $|P \xor P_i| \ge q$ for $1 \le i \le r$.
    Moreover, the algorithm finds such a path (if it exists) within the same running time bound.
\end{lemma}

\begin{lemma}\label{lem:paths-in-a-ball}
    Given non-negative integers $d, q, r$ and a path $P \in \mathcal P(s, t)$, there is an algorithm with running time $2^{O(qr^2)}n^{O(1)}$ that decides whether there are $r$ paths $P_1, \dots, P_r \in \mathcal P(s, t)$ such that
    \begin{itemize}
        \item for $1 \le i \le r$, $|P \xor P_i| \le q$ and
        \item for $1 \le i < j \le r$, $|P_i \xor P_j| \ge d$.
    \end{itemize}
    Moreover, the algorithm finds such a set of paths (if it exists) within the same running time bound.
\end{lemma}

As a first step to solve \textsc{Dissimilar Shortest Paths}, we compute (at most) $k$ paths $P_1, \ldots, P_{k'} \in \mathcal P(s, t)$ of $D$ in a greedy manner.
To this end, we first compute an arbitrary path $P_1$ from $s$ to $t$ in $D$.
By~\Cref{obs:preprocessing}, we have $P_1 \in \mathcal P(s, t)$.
Suppose that paths $P_1, \dots, P_{i-1} \in \mathcal P(s, t)$ are computed so far.
Then we compute a path $P_{i} \in \mathcal P(s, t)$ that satisfies $|P_{i} \xor P_j| \ge 3^{k - i}d$ for all $1 \le j < i$.
Such a path $P_{i}$ can be found by the algorithm in \Cref{lem:farthest-path} with $q = 3^{k - i}d$ if it exists.
We repeat this until $k$ paths are found or no such a path is found.
If this procedure successfully computes the set of $k$ paths $P_1, \ldots, P_k \in \mathcal P(s, t)$, we are done.
Thus, in the following, we assume that there are $k'$ paths $\mathcal P = \{P_1, \dots, P_{k'}\}$ computed by this procedure for some $k' < k$.
By the construction of $\mathcal P$, these paths satisfy that $|P_i \xor P_j| \ge 3^{k - j}d$ for $1 \le i < j \le k'$.

\begin{lemma}\label{lem:uniqueness-of-a-ball}
    Let $k' < k$.
    For $P \in \mathcal P(s, t)$, there is a unique path $P_i \in \mathcal P$ such that $|P \xor P_i| < 3^{k - k' - 1}d$.
\end{lemma}
\begin{proof}
    Suppose otherwise.
    If there is no such a path $P_i \in \mathcal P$, the above greedy algorithm computes $P \in \mathcal P(s, t)$ as a $(k'+1)$th path.
    If there are two paths $P_i, P_j \in \mathcal P$ with $i < j$ such that both $|P \xor P_i| < 3^{k-k'-1}d$ and $|P \xor P_j| < 3^{k-k'-1}d$ hold.
    Then, by~\Cref{obs:triangle-inequality}, we have
    \begin{align*}
        |P_i \xor P_j| \le |P_i \xor P| + |P \xor P_j| < 2 \cdot 3^{k-k'-1}d < 3^{k-k'}d.
    \end{align*}
    However, it holds that $|P_i \xor P_j| \ge 3^{k-j}d \ge 3^{k-k'}d$, which yields a contradiction.
\end{proof}

Let $q = 3^{k-k' - 1}d$ and let $\mathcal P_i = \{P \in \mathcal P(s, t) : |P \xor P| < q\}$ for $1 \le i \le k'$.
By~\Cref{lem:uniqueness-of-a-ball}, $\{\mathcal P_1, \dots, \mathcal P_{k'}\}$ is a partition of $\mathcal P(s, t)$.
Moreover, the following lemma holds.
\begin{lemma}\label{lem:distance-between-balls}
    For $P \in \mathcal P_i$ and $P' \in \mathcal P_j$ with $i \neq j$, it holds that $|P \xor P'| \ge d$.
\end{lemma}
\begin{proof}
    Suppose $|P \xor P'| < d$.
    Then, by~\Cref{obs:triangle-inequality},
    \begin{align*}
        |P_i \xor P_j| &\le |P_i \xor P| + |P \xor P'| + |P' \xor P_j|\\
        &< q + d + q\\
        &< 3^{k - k'}d,
    \end{align*}
    which contradicts the fact that $|P_i \xor P_j| \ge 3^{k - j}d$.
\end{proof}

Let $r_1, \dots, r_{k'}$ be non-negative integers with $k = r_1 + \dots + r_{k'}$.
For $1 \le i \le k'$, we compute $r_i$ paths $P^i_1, \dots, P^i_{r_i} \in \mathcal P_i$ such that $|P^i_j \xor P^i_{j'}| \ge d$.
This can be done by the algorithm of \Cref{lem:paths-in-a-ball} by setting $P = P_i$ and $r = r_i$.
Once we have those paths $P^i_1, \dots, P^i_{r_i}$ for each $1 \le i \le k'$, by~\Cref{lem:distance-between-balls}, we conclude that there are $k$ paths $P_1, \dots, P_k$ such that $|P_i \xor P_i| \ge d$ for $1 \le i < j \le k$.
Thus, by trying all possible combinations of $r_1, \dots, r_{k'}$, we can solve \textsc{Dissimilar Shortest Paths}.

We next estimate the running time bound of the above algorithm.
We can compute the directed acyclic graph $D'$ in~\Cref{obs:preprocessing} and a path $P_1$ in polynomial time by a standard single-source shortest path algorithm.
For $2 \le i \le k$, given $\{P_1, \dots, P_{i-1}\} \subseteq \mathcal P(s, t)$, we can compute a path $P_i \in \mathcal P(s, t)$ that satisfies $|P_i \xor P_j| \ge 3^{k - i}d$ for all $1 \le j < i$ in time $2^{O(k^2\log d)}n^{O(1)}$ by the algorithm in \Cref{lem:farthest-path}.
By just repeating this at most $k$ times, we can compute the paths $\{P_1, \dots, P_{k'}\}$ in time $2^{O(k^2\log d)}n^{O(1)}$ as well.
For each combination $0 \le r_1, \dots, r_{k'} \le k$ with $k = r_1 + \dots + r_{k'}$, we compute $r_i$ paths in $\mathcal P_i$ that are ``far away'' from each other in time $2^{O(3^kdr_i^2)}n^{O(1)} \subseteq  2^{O(3^kdk^2)}n^{O(1)}$ for all $i$.
This can be done by the algorithm in \Cref{lem:paths-in-a-ball}.
Since the number of such combinations on $r_i$ is upper bounded by $k^{O(k)}$, the total running time of our algorithm is upper bounded by $2^{O(3^kdk^2)}n^{O(1)}$ as well.

We would like to note that the above strategy due to \cite{FominGPP021:STACS:Diverse} is quite versatile.
When we have two subroutines stated in \Cref{lem:farthest-path} and \Cref{lem:paths-in-a-ball},  we can find ``dissimilar'' solutions for other combinatorial problems.
In particular, it is possible to extend the strategy to other metrics defined on solutions.
However, our algorithms of \Cref{lem:farthest-path} and \Cref{lem:paths-in-a-ball} heavily rely on some properties of Hamming distance, which are hard to generalize them to other metrics.

\subsection{Proof of \texorpdfstring{\Cref{lem:farthest-path}}{}}

Let $D = (V, A)$ be a directed acyclic graph with $s, t \in V$.
Let $P_1, \dots, P_r \in \mathcal P(s, t)$ and $q \in \mathbb N$.
The goal of this subsection is to show an algorithm to compute a path $P \in \mathcal P(s, t)$ such that $|P \xor P_i| \ge q$ for $1 \le i \le r$.

Let $(v_1, \dots, v_n)$ be a topological ordering of the vertices in $D$, that is, every arc $(v_i, v_j) \in A$ satisfies $i < j$.
Due to the preprocessing phase, $s$ is the unique source vertex and $t$ is the unique sink vertex in $D$, which implies that $v_1 = s$ and $v_n = t$.
For $1 \le i \le n$, let $A_i$ be the set of arcs in $D$ that are incoming to some vertex in $\{v_1, \dots, v_i\}$ (i.e., $A_i = \{(v_j, v_{j'}) \in A: 1 \le j' \le i\}$).    
For each arc $e = (v_i, v_j) \in A$, we denote by $L(e)$ the integer vector $(\gamma'_1, \dots, \gamma'_r)$ defined as $\gamma'_k = |(P_k \cap (A_j \setminus A_i)) \xor \{e\}|$ for $1 \le k \le r$.

For $1 \le i \le n$ and for each vector $\Gamma = (\gamma_1, \dots, \gamma_r)$, our algorithm computes a path $P \in \mathcal P(s, v_i)$ such that $|P \xor (P_j \cap A_i)| \ge \gamma_j$ for $1 \le j \le r$.
We define $\optf(i, \gamma) = \True$ if and only if such a path $P$ exists. 

\begin{lemma}\label{lem:farthest-path-dp}
    For $1 \le j \le n$ and vector $\Gamma = (\gamma_1, \dots, \gamma_r)$,
    \begin{align*}
        \optf(j, \Gamma) = \begin{dcases}
            \True & \text{if } j = 1 \text{ and } \Gamma = (0, \dots, 0)\\
            \False & \text{if } j = 1 \text{ and } \Gamma \neq (0, \dots, 0)\\
            \bigvee_{e = (v_i, v_j) \in \delta^-_D(v_j)} \optf(v_i, \Gamma - L(e)) & \text{otherwise}
        \end{dcases},
    \end{align*}
    where we define $\optf(j, \Gamma) = \False$ if $\Gamma$ contains a negative component.
\end{lemma}
\begin{proof}
    The proof is done by induction on $j$.
    The base case $j = 1$ is clear from the definition.

    Suppose that $j > 1$ and $\optf(j, \Gamma) = \True$.
    Let $P \in \mathcal P(s, v_j)$ such that $|P \xor (P_k \cap A_j)| \ge \gamma_k$ for $1 \le k \le r$ and let $e = (v_i, v_j)$ be the unique arc incoming to $v_j$ in $P$.
    Let $P' = P - e$.
    Then, for $1 \le k \le r$,
    \begin{align*}
        |P \xor (P_k \cap A_j)| &= |(P' \cup \{e\}) \xor ((P_k \cap A_i) \cup (P_k \cap (A_j \setminus A_i)) |\\
        &= |P' \xor (P_k \cap A_i)| + |\{e\} \xor (P_k \cap A_i)| + |P' \xor (P_k \cap (A_j \setminus A_i))| +\\
        &\hspace{4.5cm} |\{e\} \xor (P_k \cap (A_j \setminus A_i))| - |P| - |P_k \cap A_j|,
    \end{align*}
    where the second equality follows from \Cref{obs:partition}.
    As $e \in A_j \setminus A_i$ and $P' \subseteq A_i$, we have $|\{e\} \xor (P_k \cap A_i)| = |P_k \cap A_i| + 1$ and $|P' \xor (P_k \cap (A_j\setminus A_i))| = |P'| + |P_k \cap (A_j \setminus A_i)|$.
    Thus,
    \begin{align*}
        |P \xor (P_k \cap A_j)| &= |P' \xor (P_k \cap A_i)| + |\{e\} \xor (P_k \cap (A_j \setminus A_i))| +\\
        &\hspace{2.5cm} |P_k \cap A_i| + 1 + |P'| + |P_k \cap (A_j \setminus A_i)| - |P| - |P_k \cap A_j|\\
        &= |P' \xor (P_k \cap A_i)| + |\{e\} \xor (P_k \cap (A_j \setminus A_i))| + \\
        &\hspace{2.5cm} \underbrace{1 + |P'| - |P|}_0 + \underbrace{|P_k \cap A_i| + |P_k \cap (A_j \setminus A_i)| - |P_k \cap A_j|}_0\\
        &= |P' \xor (P_k \cap A_i)| + |\{e\} \xor (P_k \cap (A_j \setminus A_i))|.
    \end{align*}
    This implies that $|P' \xor (P_k \cap A_i)| = |P \xor (P_k \cap A_i)| - \gamma'_k \ge \gamma_k - \gamma'_k$.
    Applying the same argument for all $1 \le k \le r$, we have $\optf(v_i, \Gamma - L(e)) = \True$.
    It is not hard to see that this transformation is reversible: For a path $P' \in \mathcal P(s, v_i)$ with $|P' \xor (P_k \cap A_i)|$ and $e = (v_i, v_j)$, $P' + e \in \mathcal P(s, v_j)$ and satisfies $|(P' + e) \xor (P_k \cap A_j)| = |P' \xor (P_k \cap A_i)| + \gamma'_k$.
\end{proof}

By~\Cref{lem:farthest-path-dp}, $D$ has a desired path $P \in \mathcal P(s, t)$ if and only if $\optf(n, (q, \dots, q)) = \True$.
To compute $\optf(n, (q, \dots, q))$, there are $n \cdot (q + 1)^r = 2^{O(r\log q)}n$ subproblems, each of which can be evaluated in time $2^{O(r \log q)}n^{O(1)}$ by dynamic programming.
By a standard trace back technique for dynamic programming, we can construct such a path $P$ in the same running time bound.
Hence, \Cref{lem:farthest-path} follows.

\subsection{Proof of \texorpdfstring{\Cref{lem:paths-in-a-ball}}{}}
To prove \Cref{lem:paths-in-a-ball}, we need some auxiliary definitions.
In the following, we again assume that $D$ is acyclic.

For $P, P' \in \mathcal P(s, t)$, we refer to the set of arcs in $P \xor P'$ as a \emph{bypass}.
The following lemma is immediate from the definition.

\begin{lemma}\label{lem:xor-bypass}
    For $P, P', P'' \in \mathcal P(s, t)$, let $B' = P \xor P'$ and $B'' = P \xor P''$.
    Then, $P' \xor P'' = B' \xor B''$.
\end{lemma}

In the following, we fix $P \in \mathcal P(s, t)$ and denote the vertices of $P$ by $v_1, \dots, v_\ell$ with $v_1 = s$, $v_\ell = t$, and $(v_i, v_{i+1}) \in P$ for $1 \le i < \ell$.
When we refer to a bypass $B$, we implicitly assume that $B = P \xor P'$ for some $P' \in \mathcal P(s, t)$.
Let $P' \in \mathcal P(s, t)$ and let $B = P \xor P'$.
By reversing all arcs of $P'$ in $B$, the obtained graph is Eulerian, which can be decomposed into arc-disjoint directed cycles.
We say that a bypass $B$ is \emph{minimal} if the graph obtained in this way is a (single) cycle.
In other words, $B$ consists of two internally vertex-disjoint paths from $v_i$ to $v_j$ for some $i \le j$.
Conversely, every bypass can be obtained by taking the disjoint union of minimal bypasses.
See~\Cref{fig:bypass} for an illustration.
\begin{figure}
    \centering
    \includegraphics{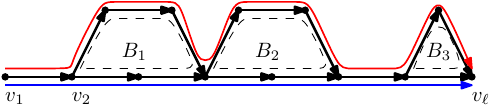}
    \caption{The figure illustrates a bypass $B$ consisting of red and blue paths from $v_1$ to $v_\ell$. The arc set of $B$ can be decomposed into three minimal bypasses $B_1$, $B_2$, and $B_3$.}
    \label{fig:bypass}
\end{figure}
This decomposition enables us to compute a bypass by dynamic programming.

Suppose that there are $r$ paths $P^*_1, \dots, P^*_r \in \mathcal P(s, t)$ that satisfy the conditions of \Cref{lem:paths-in-a-ball}, that is, $|P \xor P^*_i| \le q$ for $1 \le i \le r$ and $|P^*_i \xor P^*_j| \ge d$ for $1 \le i < j \le r$.
For each $i$, let $B^*_i = P \xor P^*_i$, and let $B^* = \bigcup_{1 \le i \le r} B^*_i$.
Obviously, we have $|B| \le qr$.
Let $m = |A|$ be the number of arcs in $D$ and let $\mathcal F$ be an $(m, qr)$-perfect hashing family.
Then, there is a function $f \in \mathcal F$ such that $f(e) \neq f(e')$ for $e, e' \in B^*$ with $e \neq e'$.
In the following, we consider the function $f$ as a (not unnecessarily proper) edge coloring of $D$, and each arc in $B^*$ receives a distinct color under $f$.
For $X \subseteq A$, we write $f(X)$ to denote the set of colors used in $X$.

\begin{lemma}\label{lem:red:coloring}
    For $1 \le i \le r$, let $f(B^*_i) = C^*_i$.
    Then, for $1 \le i < j \le r$, we have $|C^*_i \xor C^*_j| \ge d$.
\end{lemma}
\begin{proof}
    By the assumption, $|P^*_i \xor P^*_j| \ge d$.
    Then, we have
    \begin{align*}
        |C^*_i \xor C^*_j| = |f(B^*_i) \xor f(B^*_j)| = |B^*_i \xor B^*_j| = |P^*_i \xor P^*_j| \ge d,
    \end{align*}
    where the second equality follows from the fact that each arc in $B^*$ has a distinct color and the third equality follows from \Cref{lem:xor-bypass}.
\end{proof}

\begin{lemma}\label{lem:red:coloring2}
    Let $B$ and $B'$ be bypasses with $f(B) = C$ and $f(B') = C'$ for some $C, C' \subseteq [qr]$.
    Let $Q = P \xor B$ and $Q' = P \xor B'$ be paths in $\mathcal P(s, t)$.
    Then, $|Q \xor Q'| = |B \xor B'| \ge |C \xor C'|$.
\end{lemma}
\begin{proof}
    As $|B \xor B'| \ge |f(B) \xor f(B)|$, we have $|B \xor B'| \ge |C \xor C'|$.
    By~\Cref{lem:xor-bypass}, $|Q \xor Q'| = |B \xor B'| \ge |C \xor C'|$.
\end{proof}

For $C \subseteq [qr]$, a subgraph $D'$ of $D$ is said to be \emph{$C$-colorful} (under $f$) if $f(D') = C$ and $|D'| = |C|$, that is, each arc in $D'$ receives a distinct color of $C$ (under $f$).
We say that $C$ is \emph{realizable} (under $f$) if $D$ has a $C$-colorful bypass.
By~\Cref{lem:red:coloring},
\begin{itemize}
    \item bypass $B^*_i$ is $C^*_i$-colorful with $f(B^*_i) = C^*_i$ for $1 \le i \le r$ and
    \item $|C^*_i \xor C^*_j| \ge d$ for $1 \le i < j \le r$.
\end{itemize}
Conversely, suppose that there are $r$ sets of colors $C_1, \dots, C_r \subseteq [qr]$ such that $D$ has a $C_i$-colorful bypass $B_i$.
By~\Cref{lem:red:coloring2}, there are $r$ paths $Q_1, \dots, Q_r \in \mathcal P(s, t)$ such that $|Q_i \xor Q_j| \ge |C_i \xor C_j|$ for $1 \le i < j \le r$.
Thus, it suffices to compute $r$ realizable sets $C_1, \dots, C_r \subseteq [qr]$ such that $|C_i \xor C_j| \ge d$ for $1 \le i < j \le r$.

\begin{lemma}\label{lem:colorful-minimal-bypass}
    Given integers $i, i'$ with $1 \le i < i' \le \ell$, we can decide whether $D$ has a minimal $C$-colorful bypass $B$ such that $B \cap P = \{(v_j, v_{j+1}): i \le j < i'\}$ for all $C \subseteq [qr]$ with $|C| \le q$ in total time $2^{qr}n^{O(1)}$.
\end{lemma}
\begin{proof}
    Let $P[i, i'] = \{(v_j, v_{j+1}): i \le j < i'\}$ be the subpath of $P$ from $v_i$ to $v_i'$.
    If $P[i, i']$ is not $C'$-colorful for any $C' \subseteq C$, we can clearly conclude that no such a bypass exists.
    Thus, we assume that $P[i, j]$ is $C'$-colorful for some $C' \subseteq C$.
    Under this assumption, it suffices to compute a $(C \setminus C')$-colorful path from $v_i$ to $v_{i'}$ in $D' \coloneqq D[V \setminus V(P[i, i'])]$, which can be computed by the following recurrence.
    For $v \in V$ and $C'' \subseteq C \setminus C'$, $\optf(v, C'') = \True$ if and only if there is a $C''$-colorful path from $v_i$ to $v$ in $D'$.
    Then, it is easy to prove the following recurrence holds:
    \begin{align*}
        \optf(v, C'') = \begin{dcases}
            \True  & \text{if } v = v_i, C'' = \emptyset\\
            \False  & \text{if } v = v_i, C'' \neq \emptyset\\
            \bigvee_{\substack{e=(w,v) \in \delta_{D'}^-(v)\\f(e) \in C''}}\optf(w, C'' \setminus \{f(e)\}) & \text{otherwise}
        \end{dcases}.
    \end{align*}
    The recurrence can be evaluated in time
    \begin{align*}
        \sum_{0 \le i \le q} \binom{qr}{i}n^{O(1)} \subseteq 2^{qr}n^{O(1)}
    \end{align*} by dynamic programming.
\end{proof}

Using the algorithm in \Cref{lem:colorful-minimal-bypass}, we can decide if $D$ has a minimal $C$-colorful bypass $B$ such that $B \cap P = P[i, i']$ for $1 \le i < i' \le \ell$ and $C \subseteq [qr]$.
In the following, we use a Boolean predicate $\minbp$: $\minbp(i, i', C) \coloneqq \True$ if and only if $D$ has such a minimal bypass.

\begin{lemma}\label{lem:colorful-bypass-alg}
    Given an integer $i$ with $1 \le i \le \ell$, we can decide whether $D$ has a $C$-colorful bypass $B$ such that $B \cap P  \subseteq \{(v_j, v_{j + 1}) : 1 \le j < i\}$ for all $C \subseteq [qr]$ with $|C| \le q$ in total time $2^{qr}n^{O(1)}$.
\end{lemma}
\begin{proof}
    For $1 \le i \le \ell$ and $C \subseteq [qr]$, we define a Boolean predicate $\bp$ as $\bp(i, C) = \True$ if and only if $D$ has a $C$-colorful bypass $B$ such that $B \cap P \subseteq P[1, i]$, where $P[1, i] = \{(v_j, v_{j + 1}) : 1 \le j < i\}$.
    Thus, our goal is to compute $\bp(i, C)$.
    This can be done by evaluating the following recurrence:
    \begin{align*}
        \bp(i, C) = \begin{dcases}
            \True  & \text{if } C = \emptyset\\
            \False  & \text{if } i = 1, C\neq \emptyset\\
            \bigvee_{\substack{1 \le j \le j' \le i\\C' \subseteq C}}(\bp(j, C \setminus C') \land \minbp(j, j', C')) & \text{otherwise}
        \end{dcases}.
    \end{align*}
    Every $C$-colorful bypath $B$ can be decomposed into arc-disjoint minimal bypasses such that the ``rightmost'' one is $C'$-colorful.
    Moreover, from a bypass $B$ with $B \cap P \subseteq P[1, j]$ and a minimal bypass $B'$ with $B' \cap P = P[j, j']$, we can construct a bypass $B'' = B \cup B'$ with $B'' \cap P \subseteq P[1, j']$.
    Thus, the above recurrence follows.
    
    We can evaluate the recurrence within the claimed running time by dynamic programming.
\end{proof}

Now, we are ready to prove \Cref{lem:paths-in-a-ball}.
Let $D = (V, A)$ be the directed acyclic graph obtained by applying the preporcessing in \Cref{obs:preprocessing}.
Let $\mathcal F$ be an $(m, qr)$-perfect hash family with $m = |A|$.
Such a family $\mathcal F$ of size $2^{O(qr)}\log m$ can be computed in time $2^{O(qr)}m\log m$ using~\Cref{thm:perfect}.
For each $f \in \mathcal F$, we can decide whether $C$ is realizable (i.e., $\bp(\ell, C) = \True$) for all $C \subseteq [qr]$ with $|C| \le q$ in total time $2^{qr}n^{O(1)}$ under $f$.
By~\Cref{lem:red:coloring,lem:red:coloring2}, $D$ has $r$ paths $P_1, \dots, P_r$ from $s$ to $t$ such that $|P_i \xor P_j| \ge d$ for $1 \le i < j \le r$ if and only if there are $r$ sets of realizable $C_1, \dots, C_r \subseteq [qr]$ such that $|C_i \xor C_j| \ge d$ for $1 \le i < j \le r$ for some $f \in \mathcal F$.
This can be checked in time $2^{qr^2}n^{O(1)}$ by simply enumerating all combinations of $r$ realizable sets.
Therefore, the total running time of the algorithm of \Cref{lem:paths-in-a-ball} is upper bounded by $2^{O(qr^2)}n^{O(1)}$.

\section{Hardness}


\subsection{W[1]-hardness parameterized by \texorpdfstring{$k$}{}}

We perform a parameterized reduction from \textsc{Unary Bin Packing}.
In \textsc{Unary Bin Packing}, given a set of $n$ positive integers $A = \{a_1, \dots, a_n\}$ and unary-encoded integers $k, B$, the goal is to determine whether there is a partition of $A$ into $k$ sets $A_0, \dots, A_{k-1}$ such that $\sum_{a \in A_i} a = M$ for all $0 \le i <k$.
This problem is known to be W[1]-hard parameterized by~$k$~\cite{JansenKMS13:JCSS:Bin}.

From an instance of \textsc{Unary Bin Packing}, we construct an unweighted directed graph $D$ as follows.
Without loss of generality, we assume that $k \ge 2$, $\sum_{a \in A} = kM$, and $a_i < M$ for all $i$.
Moreover, we can assume that $M - a_i - 2 \ge 0$ for all $i$ by doubling each $a_i$ and $M$.
The graph $D$ consists of two parts $D_1$ and $D_2$.
In the following, we frequently use a directed path of $M-2$ arcs, which is denoted by $P$.
    
We first construct a basic building block $H$.
The graph $H$ contains $k$ copies of $P$.
We then add two vertices $u$ and $v$ to $H$ and add arcs from $u$ to the sources of the copies and arcs from the sinks of the copies to $v$.
The first part $D_1$ of $D$ is constructed by taking $2k-2$ copies $H_1, \dots, H_{2k-2}$ of $H$ and identifying $v_i$ and $u_{i+1}$ for $1 \le i \le 2k-3$, where $u_i$ and $v_i$ are the (unique) source and sink of $H_i$, respectively.
We denote the $k$ copies of $P$ in $H_i$ by $P_{i, 1}, \dots, P_{i, k}$.
We let $s = u_1$ and $r = v_{2k-2}$.

The second part $D_2$ of $D$ consists of $n$ directed acyclic graphs $H'_i$, each of which corresponds to $a_i$ of the instance of \textsc{Unary Bin Packing}.
For $1 \le i \le n$, we construct a directed graph $D_i$ as follows.
Similarly to the construction of $H$, $H'_i$ contains source $p_i$ and sink $q_i$ and $2k-2$ copies of $P$ connected to $p_i$ and $q_i$ in the same way as above.
Moreover, we add three directed paths, where one of them $Q_i$ has $a_i-1$ arcs and the other two $Q'_i$ and $Q''_i$ have $M-a_i-2$ arcs.
We connect them to $H'_i$ by adding arcs from $p_i$ to the source of $Q_i$; from the sink of $Q_i$ to each source of $Q'_i$ and $Q''_i$; from each sink of $Q'_i$ and $Q''_i$ to $q_i$.
For $1 \le i \le n-1$, we identify vertex $q_i$ with $p_{i+1}$ and vertex $r$ with $p_{1}$.
We let $t = q_{n}$.
The construction of $D$ is done, which is illustrated in~\Cref{fig:thegraph}.
Observe that each path from $s$ to $t$ has the same length $\ell \coloneqq (n + 2k-2)M$, and hence all of these paths are shortest $st$-paths.
\begin{figure}
    \centering
    \includegraphics[width=\textwidth]{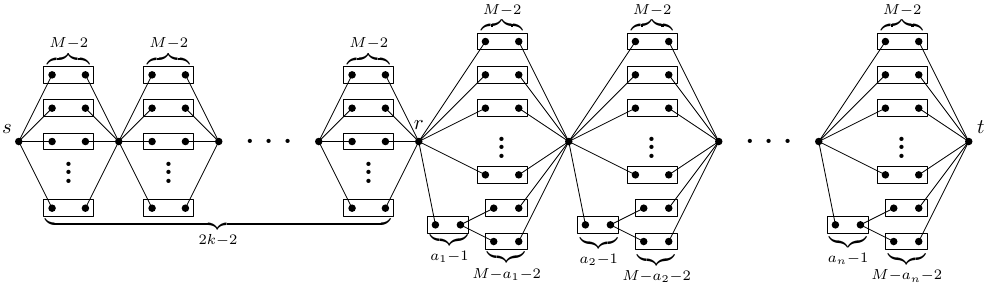}
    \caption{The figure illustrates the graph $D$. The square boxes represent directed paths of some length. All arcs are directed from left to right.}
    \label{fig:thegraph}
\end{figure}

We claim that $(A, k, M)$ is a yes-instance of \textsc{Unary Bin Packing} if and only if $D$ has $2k$ shortest paths $P_1, \dots, P_{2k}$ from $s$ to $t$ such that $|P_i \xor P_{i'}| \ge 2\ell - 2M$ for $1 \le i < i' \le 2k$.

\begin{lemma}\label{lem:w[1]-hardness:forward}
    If there is a partition $\{A_1, \dots, A_k\}$ of $A$ such that $\sum_{a \in A_i} a = B$ for all $i$, then there are $2k$ paths $P_0, \dots, P_{2k-1}$ from $s$ to $t$ in $D$ such that $|P_i \xor P_{i'}| \ge 2L - 2M$ for $1 \le i < i' \le 2k$.
\end{lemma}
\begin{proof}
    We first construct $2k$ paths $P'_0, \dots, P'_{2k-1}$ from $s$ to $r$ in $D_1$ that satisfy the following conditions: 
    \begin{itemize}\setlength{\leftskip}{0.6cm}
        \item[(P1)] for $0 \le i \le 2k-1$, $P'_i$ shares exactly $M$ arcs with $P'_{i'}$ for all $i'$ except for $i' = i + k$.
        \item[(P2)] for $0 \le i \le 2k-1$, $P'_i \cap P'_{i+k} = \emptyset$.
    \end{itemize}
    Here, the sum of indices is taken modulo $2k$.
    These paths can be constructed in the following way.
    Consider a graph $B$ with $V(B) = \{0, \dots, 2k-1\}$ that is obtained from a complete graph $K_{2k}$ of $2k$ vertices by removing a perfect matching $M_0$.
    This graph $B$ has a proper edge coloring with $2k - 2$ colors as every complete graph of $2k$ vertices has an edge coloring with $2k - 1$ colors, where each color $j$ induces a perfect matching $M_j$ of $K_{2k}$ for $0 \le j \le 2k-2$.
    From this edge coloring, we define the paths $P'_0, \dots, P'_{2k-1}$ as, by letting $M_j = \{e_1, \dots, e_k\}$, $P'_i$ and $P'_{i'}$ contain $P_{j, q}$ (as a subpath) if and only if $e_q = \{i, i'\}$, for $1 \le j \le 2k-2$.
    This construction readily implies that $|P'_i \cap P'_{i'}| = M$ if and only if $i$ is not matched to $i'$ in $M_0$.
    Moreover, $P'_i \cap P'_{i'} = \emptyset$ if $i$ is matched to $i'$ in $M_0$.
    Thus, by renaming these paths, the paths $P'_0, \dots, P_{2k-1}$ eventually satisfy (P1) and (P2).

    We next construct $2k$ paths $P''_0, \dots, P''_{2k-1}$ from $r$ to $t$ in $D_2$.
    For each $1 \le j \le n$, the subpaths of $P''_0, \dots, P''_{2k-1}$ in $H'_j$ is defined as follows.
    Suppose that $a_j$ is contained in $A_q$ for some $0 \le q \le k-1$.
    Then path $P''_{i}$ contains one of $2k-2$ copies of $P$ (as a subpath) if and only if $i \notin \{q, q + k\}$.
    Moreover, we can take these $2k - 2$ paths in such a way that they are arc-disjoint in $H'_j$.
    The path $P''_q$ (resp.\ $P''_{q + k}$) is defined as the one passing through both $Q_j$ and $Q'_j$ (resp.\ $Q_j$ and $Q''_j$) from $p_j$ to $q_j$.
    
    Now, the entire path $P_i$ is obtained by concatenating $P'_i$ and $P''_i$ for each $0 \le i \le 2k-1$.
    Let $P_i$ and $P_{i'}$ be two of these paths.
    If $i' \neq i + k$, then $P_i$ and $P_{i'}$ have $M$ common arcs in the first part $D_1$ and are arc-disjoint in the second part $D_2$, which implies that $|P_i \xor P_{i'}| = |P_i| + |P_{i'}| - 2|P_i \cap P_{i'}| = 2\ell - 2M$.
    Otherwise, $P_i$ and $P_{i'}$ are arc-disjoint in $D_1$, and have $a_j$ common arcs in $H'_j$ if $a_j \in A_i$.
    Thus, $|P_i \cap P_{i'}| = \sum_{a_j \in A_i} a_j = M$.
    Hence, $|P_i \xor P_{i'}| = 2\ell - 2M$ for $1 \le i < i' \le 2k$.
\end{proof}

To show the other direction, we need several observations on a feasible solution of \textsc{Dissimilar Shortest Paths}.
Suppose that there are $2k$ shortest paths $P_0, \dots, P_{2k-1}$ from $s$ to $t$ in $D$ such that $|P_i \xor P_{i'}| \ge 2\ell - 2M$ for $1 \le i < i' \le 2k$.
Similarly to the forward direction, each path $P_i$ can be partitioned into two subpaths $P'_i$ and $P''_i$ at $r$, where $P'_i$ and $P''_i$ belong to $D_1$ and $D_2$, respectively.

We first consider the subpaths $P'_0, \dots, P'_{2k-1}$ in the first part $D_1$.
We say that two paths $P'_i$ and $P'_{i'}$ \emph{cross} if they share an arc at some $H_{j}$, and then a pair $(\{i, i'\}, j)$ is called a \emph{crossing} (in $H_j$).
Clearly, if $P'_i$ and $P'_{i'}$ cross, we have $|P'_i \xor P'_{i''}| \ge M$.
Moreover, for $i, i'$ with $i \neq i'$, there is at most one crossing of the form $(\{i, i'\}, j)$ for some $j$ as otherwise $|P_i \xor P_{i'}| \le 2\ell-4M < 2\ell-2M$.
The following observation is easy to prove.

\begin{observation}\label{obs:w[1]-hardness:crossing}
    For $1 \le j \le 2k-2$, there are at least $k$ crossings in $H_j$.
    Moreover, there are exactly $k$ crossings in $H_j$ if and only if each $0 \le i \le 2k-1$ has exactly one index $i'\neq i$ that forms a crossing $(\{i, i'\}, j)$.
\end{observation}

This observation implies that
\begin{align*}
    \sum_{i < i'}|P'_i \cap P'_{i'}| \ge \sum_{1 \le j \le 2k-2}kM = k(2k-2)M.
\end{align*}

We next consider the subpaths $P''_0, \dots, P''_{2k-1}$ in the second part $D_2$.
Observe that
\begin{align*}
    \sum_{i < i'}|P''_i \cap P''_{i'}| \ge kM.
\end{align*}
To see this, let us consider the graph $H'_j$.
Since there are $2k - 1$ arc-disjoint paths from $p_j$ to $q_j$, at least two paths of $P''_0, \dots, P''_{2k-1}$ have a common arc in $H'_j$.
Moreover, such paths share at least $a_j$ arcs in common as $a_j < M$.
Thus, we have
\begin{align*}
     \sum_{i < i'}|P''_i \cap P''_{i'}| = \sum_{1 \le j \le n}\sum_{i < i'} |H_j \cap P''_i \cap P''_{i'}|
     \ge \sum_{1 \le j \le n} a_j
     = kM.
\end{align*}

Now, we have seen that
\begin{align*}
    \sum_{i < i'} |P_i \xor P_{i'}| &= \sum_{i < i'} (2\ell -2 |P_i \cap P_{i'}|) \\
    &= \sum_{i < i'} (2\ell - 2|P'_i \cap P'_{i'}| - 2|P''_i \cap P''_{i'}|)\\
    &\le 2k\ell(2k-1) - 2kM(2k-1)\\
    &=k(2k-1)(2\ell - 2M).
\end{align*}
As $|P_i \xor P_{i'}| \ge 2\ell-2M$ for $i, i'$ with $i \neq i'$, we have 
\begin{align*}
    \sum_{i < i'}|P_i \xor P_{i'}| \ge k(2k-1)(2\ell - 2M),
\end{align*}
yielding that $|P'_i \xor P'_{i'}| = 2\ell - 2M$ for all $0 \le i < i'\le 2k-1$.
This implies, together with \Cref{obs:w[1]-hardness:crossing}, that there are exactly $k$ crossings in $H_j$ for all $j$.
Moreover, as $a_j < M$,
\begin{itemize}\setlength{\leftskip}{0.6cm}
    \item[(P3)] there are exactly two paths $P_i$ and $P_{i'}$ with $i \neq i'$ such that $P_i$ passes through $Q_j$ and then $Q'_j$ and $P_{i'}$ passes through $Q_j$ and then $Q''_j$; Other $2k-2$ paths are arc-disjoint from any other paths in $H'_j$.
\end{itemize}

Now, we are ready to prove the following lemma.
\begin{lemma}\label{lem:w[1]-hardness:well-formed}
    The subpaths $P'_0, \dots, P'_{2k-1}$ satisfy conditions (P1) and (P2) in \Cref{lem:w[1]-hardness:forward} (by appropriately renaming them).
\end{lemma}
\begin{proof}
    Suppose for contradiction that $P'_0, \dots, P'_{2k-1}$ do not satisfy (P1).
    Since there are exactly $k$ crossings in $H_j$ for all $1 \le j \le 2k-2$, at least one path, say $P'_0$, shares an arc with $P_{i}$ for each $1 \le i \le 2k-1$.
    This implies that $|P'_0 \cap P'_{i}| = M$ for $1 \le i \le 2k-1$.
    However, as $|P'_0| = M(2k-2)$, there must be two paths $P'_i$ and $P'_{i'}$ that have crossing with $P'_0$ in $H_j$ for some $j$.
    This contradicts the fact that, by~\Cref{obs:w[1]-hardness:crossing}, there is exactly one index $i$ of the form $(\{0, i\}, j)$.
    Thus, $P'_0, \dots, P'_{2k-1}$ satisfy conditions (P1).
    By appropriately renaming them, condition (P2) also holds.
\end{proof}

Now, we assume that $P'_0, \dots, P'_{2k-1}$ satisfy (P1) and (P2).
We construct $A_i$ for $0 \le i \le k-1$ as follows.
The assumptions (P1) and (P2) imply that for $0 \le i \le 2k-1$, $P''_i$ is arc-disjoint from $P''_{i'}$ for all $i' \neq i + k$.
For each $1 \le j \le n$, there are $2k - 1$ arc-disjoint paths from $p_j$ to $q_j$ in $H'_j$.
This implies that at least two paths share an arc in $H'_j$, which must be $P''_i$ and $P''_{i + k}$ for some $0 \le i \le k-1$.
As observed in (P3), $P''_i$ passes through $Q_j$ and then $Q'_j$; $P''_{i + k}$ passes through $Q_j$ and then $Q''_j$.
We then add $a_j$ to $A_i$.
By (P3), each $a_j$ is contained in some $A_i$.
Moreover, as $|P_i \cap P_{i+k}| = M$,
\begin{align*}
    \sum_{a \in A_i} a = |P_i \cap P_{i+k}| = M,
\end{align*}
which implies that $A_0, \dots, A_{k-1}$ is a feasible solution of \textsc{Unary Bin Packing}.

\begin{lemma}\label{lem:w[1]-hardness:backward}
     If there are $2k$ paths $P_0, \dots, P_{2k-1}$ from $s$ to $t$ in $D$ such that $|P_i \xor P_{i'}| \ge 2L - 2M$ for $0 \le i < i' \le 2k-1$, then there is a partition $\{A_0, \dots, A_{k-1}\}$ of $A$ such that $\sum_{a \in A_i} a = M$ for all $i$.
\end{lemma}

Observe that the graph $D$ constructed from an instance of \textsc{Unary Bin Packing} has a linear structure.
In fact, we can show that the pathwidth of the underlying undirected graph is at most $4$.

\begin{corollary}[$\star$]
    \textsc{Dissimilar Shortest Paths} is W[1]-hard parameterized by the number of paths $k$, even on unweighted undirected graph with pathwidth at most $4$.
\end{corollary}
\begin{proof}
    In the following, we refer to $D$ as its underlying undirected graph.
    We construct a path decomposition of $D$ by concatenating path decompositions of $H_i$ for $1 \le i \le 2k-2$ and $H'_i$ for $1 \le i \le n$.
    For $H_i$, we construct a path decomposition as follows.
    If we remove $u_i$ and $v_i$ from $H_i$, the remaining graph is a disjoint union of $k$ paths, and hence it has a path decomposition of width at most $1$.
    By adding $u_i$ and $v_i$ to all bags, we obtain a path decomposition $(X^i_1, \dots, X^i_h)$ of $H_i$ of width at most~$3$ such that $u_i \in X^i_1$ and $v_i \in X^i_h$.
    As $v_i = u_{i + 1}$, the sequence of bags
    \begin{align*}
        (X^1_1, \dots ,X^1_h, X^2_1, \dots, X^2_h, \dots, X^{2k-2}_1, \dots, X^{2k-2}_h)
    \end{align*}
    is a path decomposition of the first part $D_1$.
    For $H'_i$, we can construct a path decomposition in a similar way: by removing $p_i$, $q_i$, and the unique vertex of degree $3$, we have a disjoint union of paths, and then by adding these three vertices into all bags, we have a path decomposition of $H'_i$ of width at most~$4$.
\end{proof}

\subsection{NP-hardness for constant \texorpdfstring{$d$}{}}
We show that \textsc{Dissimilar Shortest Paths} is NP-hard even when $D$ is unweighted and $d$ is constant.
To this end, we give a polynomial-time reduction from \textsc{Maximum Independent Set} on cubic graphs, which is known to be NP-complete~\cite{Garey96:TCS:Some}.
The construction is analogous to the one used in \cite{BachtlerBK22:arXiv:Almost}, where they showed that the problem of finding ``almost arc-disjoint'' paths is NP-hard even on directed acyclic graphs.
Our proof carefully adjusts their construction so that all paths in the directed graph have length exactly $\ell$ for some constant $\ell$.
To be self-contained, we give a complete reduction.

Let $H = (V, E)$ be a cubic graph.
For each $e \in E$, we construct an edge gadget $\gad(e)$ as in \Cref{fig:edge-gadget}~(Left).
\begin{figure}[t]
    \centering
    \includegraphics{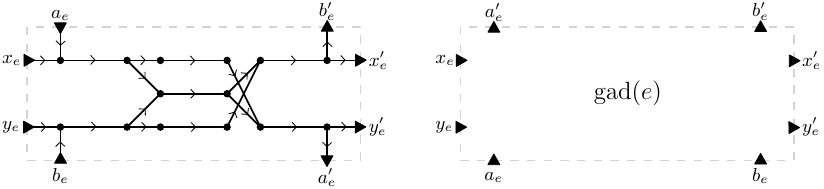}
    \caption{The left figure depicts the edge gadget $\gad(e)$. There are eight special vertices, called input and output gates, which are depicted as triangles.
    For aesthetic purposes, we rather use the right figure, which is obtained by rearranging auxiliary gates.}
    \label{fig:edge-gadget}
\end{figure}
The edge gadget has eight special vertices.
The vertices $x_e$ and $y_e$ are called \emph{vertex-input gates} and the vertices $x'_e$ and $y'_e$ are called \emph{vertex-output gates}. 
The vertices $a_e$ and $b_e$ are called \emph{auxiliary-input gates} and the vertices $a'_e$ and $b'_e$ are called \emph{auxiliary-output gates}.
The vertex-input and auxiliary-input gates are also called \emph{input gates} and the vertex-output and auxiliary-output gates are also called \emph{output gates}.
It is easy to observe that each path from an input gate to an output gate has length~$7$.
The following observation is a key property of $\gad(e)$.

\begin{observation}\label{obs:edge-gadget}
    Suppose that there are two arc-disjoint paths from auxiliary-input gates $a_e, b_e$ to auxiliary-output gates $a'_e, b'_e$.
    Then, one of the two paths, say $P_{a}$, is from $a_e$ to $a'_e$ and the other path, say $P_b$, is from $b_e$ to $b'_e$.
    Moreover, suppose that there is a path $P$ from a vertex-input gate to a vertex-output gate with $|P \cap P_a| \le 1$ and $|P \cap P_b| \le 1$.
    Then, $P$ is from $x_e$ (resp.\ from $y_e$) to $x'_e$ (resp.\ to $y'_e$).
    In particular, these three paths are uniquely determined as in~\Cref{fig:edge-gadget-path}.
    \begin{figure}
        \centering
        \includegraphics{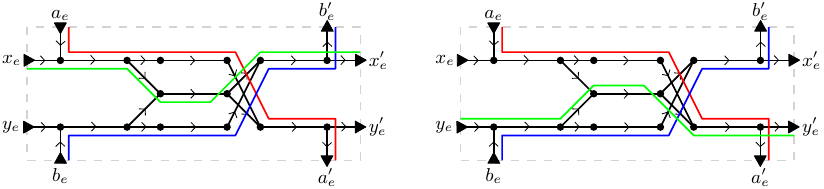}
        \caption{The three paths $P_a$, $P_v$, and $P$ are illustrated as red, blue, and green lines.}
        \label{fig:edge-gadget-path}
    \end{figure}
\end{observation}

The uniqueness of these tree paths also yields the following observation.
\begin{observation}\label{obs:edge-gadget-2}
    Suppose that there are two arc-disjoint paths $P_a$ from $a_e$ to $a'_e$ and $P_b$ from $b_e$ to $b'_e$.
    Then, there are no two arc-disjoint paths $P, P'$ from vertex-input gates $x_e, y_e$ to $x'_e, y'_e$ such that
    $|P \cap P_a| \le 1$, $|P \cap P_b| \le 1$, $|P' \cap P_a| \le $, and $|P' \cap P_b| \le 1$.
\end{observation}

For each $e \in E$, we denote by $P_{e,x}$ (resp.\ by $P_{e, y}$) the green path from $x_e$ to $x'_e$ (resp.\ from $y_e$ to $y'_e$), by $P_{e, a}$ the red path, and by $P_{e, b}$ the blue path in \Cref{fig:edge-gadget-path}.

From the graph $H$, we construct a directed graph $D$ as follows.
Let $V = \{v_1, \dots, v_n\}$.
By Vizing's theorem~\cite{Vizing64}, $H$ has a proper edge coloring with at most four colors, and such a coloring is computable in polynomial time.
We denote by $E_c$ the set of edges that are colored in color~$c$ for $1 \le c \le 4$.
The graph $D$ contains three vertices $s, t, v_V$ and an arc $(s, v_V)$.
Moreover, $D$ contains $n$ vertices corresponding to $V = \{v_1, \dots, v_n\}$ that are adjacent to $v_V$ with arcs $(v_V, v_i)$ for all $1 \le i \le n$.
To construct the remaining part, we use four \emph{layers} $L_1, \dots, L_4$, each of which corresponds to color~$c$.
For each vertex $v_i \in V$, we define a $v_it$-path such that for $1 \le c \le 4$,
\begin{itemize}
    \item if there is an edge $e \in E_c$ incident to $v_i$, the path enters $\gad(e)$ at a vertex-input gate $x_e$ (or $y_e$), passes through $P_{e, x}$ (or $P_{e, y}$), and then leaves at the corresponding vertex-output gate $x'_e$ (or $y'_e$) in $L_c$;
    \item otherwise, the path contains a subpath of length~$7$ in $L_c$.
\end{itemize}
The construction of the path is well-defined since $v_i$ has no two incident edges in $E_c$.
\Cref{fig:partial} illustrates a partial construction of the paths for all $v_i \in V$.
\begin{figure}
    \centering
    \includegraphics{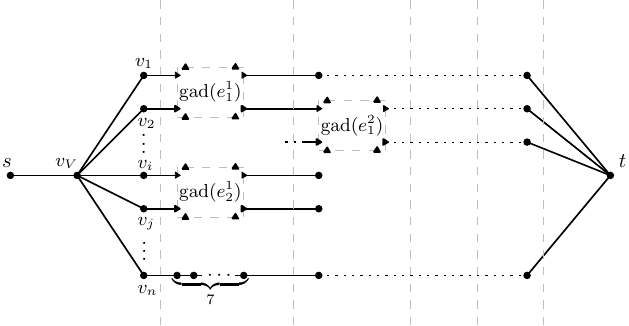}
    \caption{A partial construction of the graph $D$. There are four layers (separated by dotted lines).
        If $E_c$ contains an edge $e = \{v_i, v_j\}$, the layer $L_c$ contains the edge gadget $\gad(e)$ and the paths corresponding to $v_i$ and $v_j$ enter $\gad(e)$ at its vertex-input gates and leave at its vertex-output gates.
        All arcs are directed from left to right.}
    \label{fig:partial}
\end{figure}
By concatenating arcs $(s, v_V)$ and $(v_V, v_i)$ to this path, we have an $st$-path, which is denoted by $Q_i$.
We refer to the path $Q_i$ as a \emph{primal vertex path}.
Let us note that for two non-adjacent vertices $v_i$ and $v_j$ in $H$, the corresponding primal vertex paths $Q_i$ and $Q_j$ share only one arc $(s, v_V)$ in common.
Furthermore, all the $st$-paths have the same length $\ell = 35$.
Each path passing through $v_V$ from $s$ to $t$ is called a \emph{vertex path}.

To complete the construction, we also add a vertex $v_E$ and an arc $(s, v_E)$.
For each $e \in E$, we add two paths from $v_E$ to the auxiliary-input gates of $\gad(e)$ and two paths from the auxiliary-output gates to $t$.
The length of these paths are determined in such a way that all the $st$-paths in $D$ have the same length $\ell$.
The entire graph $D$ is illustrated in \Cref{fig:entire}.
\begin{figure}
    \centering
    \includegraphics{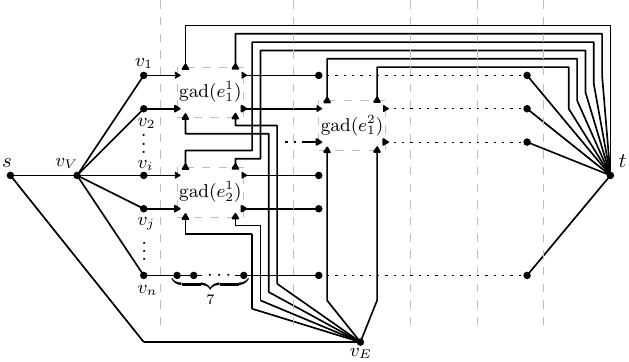}
    \caption{The figure depicts the complete picture of the graph $D$.}
    \label{fig:entire}
\end{figure}
We denote $Q_{e, a} = (s, v_E) + P_{e, a}$ and $Q_{e, b} = (s, v_E) + P_{e, b}$.
We refer to these paths $Q_{e, a}, Q_{e, b}$ as \emph{primal auxiliary paths}.
Similarly, each path passing through $v_E$ from $s$ to $t$ is called an \emph{auxiliary path}.

Now, we show that $H$ has an independent set of size at least $k$ if and only if $D$ has $k + 2|E|$ paths $P_1, \dots, P_{k + 2|E|}$ from $s$ to $t$ such that $|P_i \xor P_j| \ge 2\ell - 2$ for $1 \le i < j \le k + 2|E|$.

\begin{lemma}\label{lem:np-hardness:forward}
    Suppose that $H$ has an independent set of size at least $k$.
    Then, $D$ has $k + 2|E|$ paths $P_1, \dots, P_{k + 2|E|}$ from $s$ to $t$ such that $|P_i \xor P_j| \ge 2\ell - 2$ for $1 \le i < j \le k + 2|E|$.
\end{lemma}
\begin{proof}
    Let $X \subseteq V$ be an independent set of size $k$.
    Then, we construct a set $\mathcal P$ of $k + 2|E|$ paths as $\mathcal P = \{Q_i : v_i \in H\} \cup \{Q_{e, a}, Q_{e, b} : e \in E\}$.
    For distinct primal auxiliary paths $P, P' \in \mathcal P$, we have $P \cap P' = \{(s, v_E)\}$, and hence $|P \xor P'| \ge 2\ell - 2|P \cap P'| = 2\ell - 2$.
    For a primal vertex path $P \in \mathcal P_X$ and a primal auxiliary path $P' \in \mathcal P_E$, $P \cap P' \neq \emptyset$ if and only if the edge $e$ corresponding to the primal auxiliary path $P'$ is incident to the vertex $v_i$ corresponding to $P$.
    Moreover, if $P \cap P’ \neq \emptyset$, they share only one common edge in $\gad(e)$, as observed in \Cref{obs:edge-gadget}.
    Thus, we have $|P \xor P'| \ge 2\ell - 2$.
    For distinct primal vertex paths $P, P' \in \mathcal P$, they share only one arc $(s, v_V)$ in common as the corresponding vertices are not adjacent in $H$.
    Therefore, for distinct $P, P \in \mathcal P$, we have $|P \xor P'| \ge 2\ell - 2$.
\end{proof}

\begin{lemma}\label{lem:np-hardness:backward}
    Suppose that $D$ has $k + 2|E|$ paths $P_1, \dots, P_{k + 2|E|}$ from $s$ to $t$ such that $|P_i \xor P_j| \ge 2\ell - 2$ for $1 \le i < j \le k + 2|E|$.
    Then, $H$ has an independent set of size at least~$k$.
\end{lemma}

To prove \Cref{lem:np-hardness:backward}, we need several observations on diverse shortest $st$-paths from in $D$.
From now on, we refer to the set of paths $\mathcal P = \{P_1, \dots, P_{q}\}$ from $s$ to $t$ in $D$ satisfying $|P_i \xor P_j| \ge 2\ell - 2$ for $1 \le i < j \le q$ as a \emph{feasible set}.
Let $\mathcal P = \{P_1, \dots, P_q\}$ be a feasible set of size at least $k + 2|E|$.
Without loss of generality, we assume that $D$ has no feasible set of size at least $q + 1$.
Moreover, we assume that $\mathcal P$ satisfies the following maximality condition:
\begin{itemize} \setlength{\leftskip}{0.2cm}
    \item[(i)] $\mathcal P$ maximizes the number of auxiliary paths among all feasible sets of size $q$.
\end{itemize} 
The following observations are essentially made in \cite{BachtlerBK22:arXiv:Almost}, whose proofs are deferred to Appendix.

\begin{observation}[$\star$]\label{obs:proper-auximilary-path}
    Each auxiliary path in $\mathcal P$ intersects with exactly one edge gadget.
\end{observation}

The above observation ensures that each path in $\mathcal P$ passing through vertex-input gate is a vertex path.
In addition to~(i), we assume that $\mathcal P$ satisfies the following condition:
\begin{itemize} \setlength{\leftskip}{0.3cm}
    \item[(ii)] $\mathcal P$ maximizes the number of auxiliary paths leaving the unique intersecting edge gadget at its auxiliary-output gate among all feasible sets of size $q$ satisfying (i).
\end{itemize}

\begin{observation}[$\star$]\label{obs:leaving-auxiliary-output-gate}
    Each auxiliary path in $\mathcal P$ leaves the unique intersecting edge gadget at its auxiliary-output.
\end{observation}

\begin{observation}[$\star$]\label{obs:auxiliary-paths}
    The set $\mathcal P$ contains exactly $2|E|$ auxiliary paths.
\end{observation}

\begin{observation}[$\star$]\label{obs:primal-auxiliary-paths}
    Every auxiliary path in $\mathcal P$ is a primal auxiliary path.
\end{observation}

By~\Cref{obs:auxiliary-paths}, there are at least $k$ vertex paths in $\mathcal P$.
Moreover, by~\Cref{obs:edge-gadget,obs:primal-auxiliary-paths}, these paths must be primal vertex paths as there are two arc-disjoint subpaths from the auxiliary-input gates to the auxiliary output-gates in each edge gadget.
Let $X \subseteq V$ be the set of $q - 2|E|$ vertices of $H$ corresponding to these vertex paths.
By~\Cref{obs:edge-gadget-2}, there is no pair of two primal vertex paths entering a single edge gadget.
This implies that $X$ is an independent set of $H$, completing the proof of \Cref{lem:np-hardness:backward}.

\bibliography{main}

\begin{thebibliography}{10}

\bibitem{Akgun:finding:2000}
Vedat Akg{\"{u}}n, Erhan Erkut, and Rajan Batta.
\newblock On finding dissimilar paths.
\newblock {\em Eur. J. Oper. Res.}, 121(2):232--246, 2000.

\bibitem{AlonYZ95}
Noga Alon, Raphael Yuster, and Uri Zwick.
\newblock Color-coding.
\newblock {\em J. {ACM}}, 42(4):844--856, 1995.
\newblock \href {https://doi.org/10.1145/210332.210337}
  {\path{doi:10.1145/210332.210337}}.

\bibitem{BachtlerBK22:arXiv:Almost}
Oliver Bachtler, Tim Bergner, and Sven~O. Krumke.
\newblock Almost disjoint paths and separating by forbidden pairs, 2022.
\newblock \href {https://arxiv.org/abs/2202.10090} {\path{arXiv:2202.10090}}.

\bibitem{BasteFJMOPR22:AI:Diversity}
Julien Baste, Michael~R. Fellows, Lars Jaffke, Tom{\'{a}}s Masar{\'i}k, Mateus
  de~Oliveira~Oliveira, Geevarghese Philip, and Frances~A. Rosamond.
\newblock Diversity of solutions: An exploration through the lens of
  fixed-parameter tractability theory.
\newblock {\em Artif. Intell.}, 303:103644, 2022.
\newblock \href {https://doi.org/10.1016/j.artint.2021.103644}
  {\path{doi:10.1016/j.artint.2021.103644}}.

\bibitem{BasteJMPR19:Algorithms:FPT}
Julien Baste, Lars Jaffke, Tom{\'{a}}s Masar{\'i}k, Geevarghese Philip, and
  G{\"{u}}nter Rote.
\newblock {FPT} algorithms for diverse collections of hitting sets.
\newblock {\em Algorithms}, 12(12):254, 2019.
\newblock \href {https://doi.org/10.3390/a12120254}
  {\path{doi:10.3390/a12120254}}.

\bibitem{ChondrogiannisB20:VLDB:Finding}
Theodoros Chondrogiannis, Panagiotis Bouros, Johann Gamper, Ulf Leser, and
  David~B. Blumenthal.
\newblock Finding k-shortest paths with limited overlap.
\newblock {\em {VLDB} J.}, 29(5):1023--1047, 2020.
\newblock \href {https://doi.org/10.1007/s00778-020-00604-x}
  {\path{doi:10.1007/s00778-020-00604-x}}.

\bibitem{deBergMS23:arXiv:Finding}
Mark de~Berg, Andr{\'{e}}s~L{\'{o}}pez Mart{\'{\i}}nez, and Frits C.~R.
  Spieksma.
\newblock Finding diverse minimum s-t cuts.
\newblock In Satoru Iwata and Naonori Kakimura, editors, {\em 34th
  International Symposium on Algorithms and Computation, {ISAAC} 2023, December
  3-6, 2023, Kyoto, Japan}, volume 283 of {\em LIPIcs}, pages 24:1--24:17.
  Schloss Dagstuhl - Leibniz-Zentrum f{\"{u}}r Informatik, 2023.
\newblock \href {https://doi.org/10.4230/LIPICS.ISAAC.2023.24}
  {\path{doi:10.4230/LIPICS.ISAAC.2023.24}}.

\bibitem{DoGN023}
Anh~Viet Do, Mingyu Guo, Aneta Neumann, and Frank Neumann.
\newblock Diverse approximations for monotone submodular maximization problems
  with a matroid constraint.
\newblock In {\em Proceedings of the Thirty-Second International Joint
  Conference on Artificial Intelligence, {IJCAI} 2023, 19th-25th August 2023,
  Macao, SAR, China}, pages 5558--5566. ijcai.org, 2023.
\newblock URL: \url{https://doi.org/10.24963/ijcai.2023/617}, \href
  {https://doi.org/10.24963/IJCAI.2023/617}
  {\path{doi:10.24963/IJCAI.2023/617}}.

\bibitem{EibenLW23:arXiv:Determinantal}
Eduard Eiben, Tomohiro Koana, and Magnus Wahlström.
\newblock Determinantal sieving.
\newblock In {\em Proceedings of the 2024 Annual ACM-SIAM Symposium on Discrete
  Algorithms (SODA)}, pages 377--423, 2024.
\newblock \href {https://doi.org/10.1137/1.9781611977912.16}
  {\path{doi:10.1137/1.9781611977912.16}}.

\bibitem{FominGJPS20}
Fedor~V. Fomin, Petr~A. Golovach, Lars Jaffke, Geevarghese Philip, and Danil
  Sagunov.
\newblock Diverse pairs of matchings.
\newblock In Yixin Cao, Siu{-}Wing Cheng, and Minming Li, editors, {\em 31st
  International Symposium on Algorithms and Computation, {ISAAC} 2020, December
  14-18, 2020, Hong Kong, China (Virtual Conference)}, volume 181 of {\em
  LIPIcs}, pages 26:1--26:12. Schloss Dagstuhl - Leibniz-Zentrum f{\"{u}}r
  Informatik, 2020.
\newblock \href {https://doi.org/10.4230/LIPICS.ISAAC.2020.26}
  {\path{doi:10.4230/LIPICS.ISAAC.2020.26}}.

\bibitem{FominGPP021:STACS:Diverse}
Fedor~V. Fomin, Petr~A. Golovach, Fahad Panolan, Geevarghese Philip, and Saket
  Saurabh.
\newblock Diverse collections in matroids and graphs.
\newblock In Markus Bl{\"{a}}ser and Benjamin Monmege, editors, {\em 38th
  International Symposium on Theoretical Aspects of Computer Science, {STACS}
  2021}, volume 187 of {\em LIPIcs}, pages 31:1--31:14. Schloss Dagstuhl -
  Leibniz-Zentrum f{\"{u}}r Informatik, 2021.
\newblock \href {https://doi.org/10.4230/LIPIcs.STACS.2021.31}
  {\path{doi:10.4230/LIPIcs.STACS.2021.31}}.

\bibitem{GaoGMKTTY22}
Jie Gao, Mayank Goswami, {Karthik {C. S.}}, Meng{-}Tsung Tsai, Shih{-}Yu Tsai,
  and Hao{-}Tsung Yang.
\newblock Obtaining approximately optimal and diverse solutions via dispersion.
\newblock In Armando Casta{\~{n}}eda and Francisco
  Rodr{\'{\i}}guez{-}Henr{\'{\i}}quez, editors, {\em {LATIN} 2022: Theoretical
  Informatics - 15th Latin American Symposium, Guanajuato, Mexico, November
  7-11, 2022, Proceedings}, volume 13568 of {\em Lecture Notes in Computer
  Science}, pages 222--239. Springer, 2022.
\newblock \href {https://doi.org/10.1007/978-3-031-20624-5_14}
  {\path{doi:10.1007/978-3-031-20624-5_14}}.

\bibitem{Garey96:TCS:Some}
M.R. Garey, D.S. Johnson, and L.~Stockmeyer.
\newblock Some simplified np-complete graph problems.
\newblock {\em Theor. Comput. Sci.}, 1(3):237--267, 1976.
\newblock \href {https://doi.org/10.1016/0304-3975(76)90059-1}
  {\path{doi:10.1016/0304-3975(76)90059-1}}.

\bibitem{HackerBCA21SIGSPATIAL:Most}
Christian H{\"{a}}cker, Panagiotis Bouros, Theodoros Chondrogiannis, and Ernst
  Althaus.
\newblock Most diverse near-shortest paths.
\newblock In Xiaofeng Meng, Fusheng Wang, Chang{-}Tien Lu, Yan Huang, Shashi
  Shekhar, and Xing Xie, editors, {\em {SIGSPATIAL} '21: 29th International
  Conference on Advances in Geographic Information Systems 2021}, pages
  229--239. {ACM}, 2021.
\newblock \href {https://doi.org/10.1145/3474717.3483955}
  {\path{doi:10.1145/3474717.3483955}}.

\bibitem{HanakaKKKKO22:arXiv:Framework}
Tesshu Hanaka, Masashi Kiyomi, Yasuaki Kobayashi, Yusuke Kobayashi, Kazuhiro
  Kurita, and Yota Otachi.
\newblock A framework to design approximation algorithms for finding diverse
  solutions in combinatorial problems.
\newblock In Brian Williams, Yiling Chen, and Jennifer Neville, editors, {\em
  Thirty-Seventh {AAAI} Conference on Artificial Intelligence, {AAAI} 2023},
  pages 3968--3976. {AAAI} Press, 2023.
\newblock URL: \url{https://doi.org/10.1609/aaai.v37i4.25511}, \href
  {https://doi.org/10.1609/AAAI.V37I4.25511}
  {\path{doi:10.1609/AAAI.V37I4.25511}}.

\bibitem{HanakaKKLO22:AAAI:Computing}
Tesshu Hanaka, Yasuaki Kobayashi, Kazuhiro Kurita, See~Woo Lee, and Yota
  Otachi.
\newblock Computing diverse shortest paths efficiently: {A} theoretical and
  experimental study.
\newblock In {\em Thirty-Sixth {AAAI} Conference on Artificial Intelligence,
  {AAAI} 2022}, pages 3758--3766. {AAAI} Press, 2022.
\newblock \href {https://doi.org/10.1609/AAAI.V36I4.20290}
  {\path{doi:10.1609/AAAI.V36I4.20290}}.

\bibitem{HanakaKKO21:AAAI:Finding}
Tesshu Hanaka, Yasuaki Kobayashi, Kazuhiro Kurita, and Yota Otachi.
\newblock Finding diverse trees, paths, and more.
\newblock In {\em Thirty-Fifth {AAAI} Conference on Artificial Intelligence,
  {AAAI} 2021}, pages 3778--3786. {AAAI} Press, 2021.
\newblock \href {https://doi.org/10.1609/AAAI.V35I5.16495}
  {\path{doi:10.1609/AAAI.V35I5.16495}}.

\bibitem{Holyer81:SICOMP:NP}
Ian Holyer.
\newblock The {NP}-completeness of edge-coloring.
\newblock {\em {SIAM} J. Comput.}, 10(4):718--720, 1981.
\newblock \href {https://doi.org/10.1137/0210055} {\path{doi:10.1137/0210055}}.

\bibitem{JansenKMS13:JCSS:Bin}
Klaus Jansen, Stefan Kratsch, D{\'{a}}niel Marx, and Ildik{\'{o}} Schlotter.
\newblock Bin packing with fixed number of bins revisited.
\newblock {\em J. Comput. Syst. Sci.}, 79(1):39--49, 2013.
\newblock URL: \url{https://doi.org/10.1016/j.jcss.2012.04.004}, \href
  {https://doi.org/10.1016/J.JCSS.2012.04.004}
  {\path{doi:10.1016/J.JCSS.2012.04.004}}.

\bibitem{LiMS90}
Chung{-}Lun Li, S.~Thomas McCormick, and David Simchi{-}Levi.
\newblock The complexity of finding two disjoint paths with min-max objective
  function.
\newblock {\em Discret. Appl. Math.}, 26(1):105--115, 1990.
\newblock \href {https://doi.org/10.1016/0166-218X(90)90024-7}
  {\path{doi:10.1016/0166-218X(90)90024-7}}.

\bibitem{Liu:Finding:2018}
Huiping Liu, Cheqing Jin, Bin Yang, and Aoying Zhou.
\newblock Finding top-k shortest paths with diversity.
\newblock {\em {IEEE} Trans. Knowl. Data Eng.}, 30(3):488--502, 2018.
\newblock \href {https://doi.org/10.1109/TKDE.2017.2773492}
  {\path{doi:10.1109/TKDE.2017.2773492}}.

\bibitem{MerklPS23:ICDT:Diversity}
Timo~Camillo Merkl, Reinhard Pichler, and Sebastian Skritek.
\newblock Diversity of answers to conjunctive queries.
\newblock In Floris Geerts and Brecht Vandevoort, editors, {\em 26th
  International Conference on Database Theory, {ICDT} 2023}, volume 255 of {\em
  LIPIcs}, pages 10:1--10:19. Schloss Dagstuhl - Leibniz-Zentrum f{\"{u}}r
  Informatik, 2023.
\newblock \href {https://doi.org/10.4230/LIPIcs.ICDT.2023.10}
  {\path{doi:10.4230/LIPIcs.ICDT.2023.10}}.

\bibitem{NaorSS95:FOCS:Splitters}
Moni Naor, Leonard~J. Schulman, and Aravind Srinivasan.
\newblock Splitters and near-optimal derandomization.
\newblock In {\em 36th Annual Symposium on Foundations of Computer Science,
  Milwaukee, Wisconsin, USA, 23-25 October 1995}, pages 182--191. {IEEE}
  Computer Society, 1995.
\newblock \href {https://doi.org/10.1109/SFCS.1995.492475}
  {\path{doi:10.1109/SFCS.1995.492475}}.

\bibitem{Vizing64}
Vadim~G. Vizing.
\newblock On an estimate of the chromatic class of a $p$-graph.
\newblock {\em Discret Analiz}, 3:25--30, 1964.

\bibitem{Voss:heuristic:2015}
Caleb Voss, Mark Moll, and Lydia~E. Kavraki.
\newblock A heuristic approach to finding diverse short paths.
\newblock In {\em {ICRA}}, pages 4173--4179. {IEEE}, 2015.

\bibitem{ZhengYYW07:INFOCOM:Finding}
Si{-}Qing Zheng, Bing Yang, Mei Yang, and Jianping Wang.
\newblock Finding minimum-cost paths with minimum sharability.
\newblock In {\em {INFOCOM} 2007. 26th {IEEE} International Conference on
  Computer Communications, Joint Conference of the {IEEE} Computer and
  Communications Societies}, pages 1532--1540. {IEEE}, 2007.
\newblock \href {https://doi.org/10.1109/INFCOM.2007.180}
  {\path{doi:10.1109/INFCOM.2007.180}}.

\end{thebibliography}

\newpage
\appendix
\section{Omitted proofs}

\setcounter{theorem}{26}

\begin{observation}[restated]
    Each auxiliary path in $\mathcal P$ intersects with exactly one edge gadget.
\end{observation}
\begin{proof}
    Suppose otherwise.
    Let $P \in \mathcal P$ be an auxiliary path that intersects with $\gad(e)$ and $\gad(f)$ for distinct $e, f \in E$.
    Without loss of generality, we assume that $P$ leaves $\gad(e)$ at $x'_e$ and then enters $\gad(f)$ at $x_f$ without intersecting any other edge gadget.
    Since each auxiliary path has only common arc $(s, v_E)$ with other auxiliary paths, $\mathcal P$ has no auxiliary path passing through $b'_e$ as well as auxiliary path passing through $a_f$.
    We let $w_e$ and $w_f$ be the vertices as in \Cref{fig:proper-path}.
    \begin{figure}
        \centering
        \includegraphics{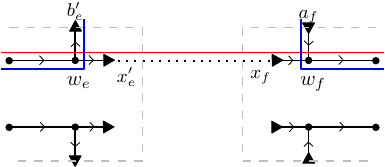}
        \caption{An illustration of two edge gadgets. The auxiliary path $P$ is depicted as the red line and two auxiliary paths $P_1$ and $P_2$ are depicted as the blue lines.}
        \label{fig:proper-path}
    \end{figure}
    We construct two auxiliary paths $P_1$ and $P_2$ as: $P_1$ passes from $s$ to $w_e$ along $P$, then leaves $\gad(e)$ at $b'_e$, and directly reaches to $t$; $P_2$ passes through $v_E$, then enters $\gad(f)$ at $a_f$, and passes from $w_f$ to $t$ along $P$.
    For any vertex path $Q \in \mathcal P$, we have $P_1 \cap Q \subseteq P \cap Q$ and $P_2 \cap Q \subseteq P \cap Q$ except for the case that $Q$ passes through $b'_e$.
    Thus, by removing such a exceptional vertex path $Q$ passing through $b'_e$ (if it exists), $(\mathcal P \cup \{P_1, P_2\}) \setminus \{P, Q\}$ is a feasible set of size at least $q$, which contradicts the maximality condition~(i).
\end{proof}


\begin{observation}[restated]
    Each auxiliary path in $\mathcal P$ leaves the unique intersecting edge gadget at its auxiliary-output.
\end{observation}
\begin{proof}
    Suppose otherwise.
    Let $P$ be an auxiliary path in $\mathcal P$ that intersects $\gad(e)$ for some $e \in E$ and leaves it from its vertex-output gate, say $x'_e$.
    By~\Cref{obs:proper-auximilary-path}, $P$ does not intersects with any other edge gadget.
    Moreover, there is no other path in $\mathcal P$ passing through $x'_e$, as otherwise, such a path has at least two common arcs with $P$.
    By letting $w_e$ as in~\Cref{fig:proper-path}, we can reroute $P$ as follows.
    If $\mathcal P$ has no path passing through $w_e$ other than $P$, we simply replace the subpath $P'$ of $P$ starting from $w_e$ with the path $Q'$ from $w_e$ to $t$ passing through $b'_e$.
    Otherwise, that is, $\mathcal P$ has a path $Q$ that contains $Q'$ as a subpath, we ``exchange'' their subpaths $P'$ and $Q'$.
    Since there are at least two arcs in $Q'$, $Q$ is the unique path in $\mathcal P$ containing $Q'$ as a subpath.
    Thus, by exchanging these subpaths, we have a new feasible set that has more auxiliary paths leaving their own unique intersecting edge gadget at its auxiliary-output gate, which contradicts (ii).
\end{proof} 

\begin{observation}[restated]
    The set $\mathcal P$ contains exactly $2|E|$ auxiliary paths.
\end{observation}
\begin{proof}
    As all auxiliary paths have $(s, v_E)$ in common, $\mathcal P$ can have at most $2|E|$ auxiliary paths.
    In the following, we show that $\mathcal P$ contains at least $2|E|$ auxiliary paths.
    
    Suppose otherwise.
    Since there are at most $2|E| - 1$ auxiliary paths in $\mathcal P$, there is an edge gadget $\gad(e)$ such that at least one auxiliary-input gate is not used in any path in $\mathcal P$.
    We assume without loss of generality that $a_e$ is such an unused auxiliary-input gate.
    We distinguish several cases.

    \textbf{Case I: $\mathcal P$ has no vertex path passing through $x_e$.}
    Suppose moreover that $\mathcal P$ has no vertex path passing through $y_e$.
    If $\mathcal P$ has no auxiliary path passing through $b_e$ and leaving $\gad(e)$ at $a'_e$, then the primal auxiliary path $Q_{e, a}$ does not has a common edge of $\gad(e)$ with other paths in $\mathcal P$.
    Thus, $\mathcal P \cup \{Q_{e, a}\}$ is a feasible set of size at least $q + 1$, contradicting the choice of $\mathcal P$.
    Otherwise, $\mathcal P$ has an auxiliary path $P$ passing through $b_e$ and leaving $\gad(e)$ at $a'_e$ as in~\Cref{fig:auxiliary-paths-proof1}~(Left).
    \begin{figure}
        \centering
        \includegraphics{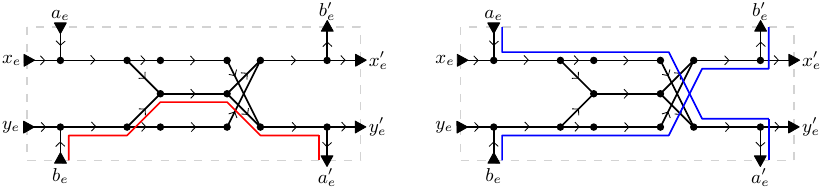}
        \caption{The left figure illustrates an auxiliary path passing through $b_e$ and leaving $\gad(e)$ at $a'_e$ or $y'_e$.
        We can replace the red path with two auxiliary paths as in the right figure.}
        \label{fig:auxiliary-paths-proof1}
    \end{figure}
    Since $\mathcal P$ has no other paths intersecting with $\gad(e)$, $\mathcal P' = (\mathcal P \setminus \{P\}) \cup \{Q_{e, a}, Q_{e, b}\}$ is a feasible set, which has more paths than $\mathcal P$, a contradiction.
    If $\mathcal P$ has a vertex path $Q$ passing through $y_e$, we apply the same argument to $\mathcal P \setminus \{Q\}$, which yields a contradiction (with respect to~(i)).

    \textbf{Case II: $\mathcal P$ has a vertex path $P$ passing through $x_e$ and leaves $\gad(e)$ at $b'_e$ or $x'_e$.}
    In this case, $\mathcal P$ has no vertex path passing through $y_e$.
    Thus, by applying the argument used in Case I to $\mathcal P \setminus \{P\}$, we have a contradiction (with respect to~(i)).

    \textbf{Case III: $\mathcal P$ has a vertex path $P$ passing through $x_e$ and leaves $\gad(e)$ at $a'_e$ or $y'_e$.}
    If $\mathcal P$ has no vertex path passing through $y_e$, we can apply the same argument as Case II.
    Thus, we assume that there is a vertex path $Q \in \mathcal P$ passing through $y_e$.
    As $P$ and $Q$ have a common arc $(s, v_V)$, they are arc-disjoint in $\gad(e)$, meaning that $Q$ passes through $b'_e$ or $x'_e$.
    Considering symmetric cases, there are essentially two cases (Left) and (Right) as in~\Cref{fig:auxiliary-paths-proof2}.   
    \begin{figure}
        \centering
        \includegraphics{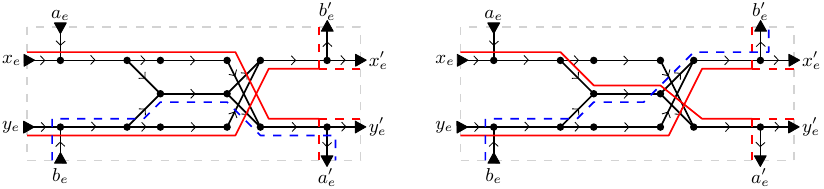}
        \caption{There are two possible cases (Left) and (Right). There can be an auxiliary path $R$ passing through $b_e$, depicted by the blue dashed line.}
        \label{fig:auxiliary-paths-proof2}
    \end{figure}
    If $\mathcal P$ has no auxiliary path passing through $b_e$, we can construct a feasible sat $(\mathcal P \setminus \{P, Q\}) \cup \{Q_{e, a}, Q_{e,b}\}$, which contradicts the maximality~(ii).
    Suppose otherwise, that is, $\mathcal P$ has an auxiliary path $R$ passing through $b_e$.
    By~\Cref{obs:leaving-auxiliary-output-gate}, the auxiliary path $R$ leaves $\gad(e)$ at $a'_e$ as in~\Cref{fig:auxiliary-paths-proof2}~(Left) or at $b'_e$ as in~\Cref{fig:auxiliary-paths-proof2}~(Right).
    In both cases, the vertex path passing through $x_e$ leaves $\gad(e)$ at $y'_e$ (Left), and the vertex path passing through $y_e$ leaves $\gad(e)$ at $x'_e$ (Right), respectively.
    \begin{figure}
        \centering
        \includegraphics{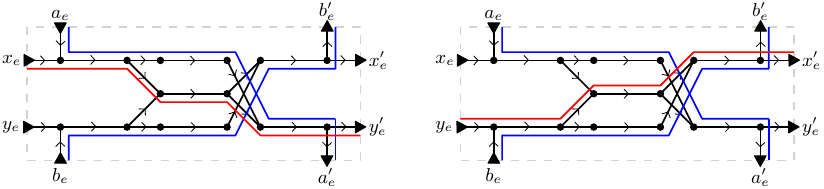}
        \caption{The figure illustrates the replacement of two cases (Left) and (Right) in \Cref{fig:auxiliary-paths-proof2}.}
        \label{fig:auxiliary-paths-proof3}
    \end{figure}
    Let $P'$ (resp.\ $Q'$) be the vertex path obtained from $P$ (resp.\ $Q$) by locally replacing the subpath inside $\gad(e)$ with the red path illustrated in the left (resp. right) figure of \Cref{fig:auxiliary-paths-proof3}.
    Thus, we construct a feasible set from $(\mathcal P \setminus \{P, Q, R\}) \cup \{Q_{e,a}, Q_{e, b}\}$ by appropriately adding $P'$ or $Q'$, depending on the cases.
    The obtained feasible set contains more auxiliary paths, contradicting the maximality~(i).
\end{proof}

\begin{observation}[restated]
    Every auxiliary path in $\mathcal P$ is a primal auxiliary path.
\end{observation}
\begin{proof}
    Let $P \in \mathcal P$ be an auxiliary path.
    By~\Cref{obs:proper-auximilary-path,obs:leaving-auxiliary-output-gate}, there is a unique edge gadget $\gad(e)$ such that $P$ enters it at an auxiliary-input gate and leaves it at an auxiliary-output gate.
    Since there are $2|E|$ auxiliary paths in $\mathcal P$, there is another auxiliary path $Q \in \mathcal P$ that enters $\gad(e)$ at the other auxiliary-input and leaves $\gad(e)$ at the other auxiliary-output gate.
    Thus, by~\Cref{obs:edge-gadget}, $P$ passes through $a_e$ (resp.\ $b_e$) and then $a'_e$ (resp.\ $b'_e$), which means $P = Q_{e, a}$ (resp.\ $Q_{e, b}$).
    Hence, $P$ is a primal auxiliary path.
\end{proof}

\end{document}